\documentclass[aps, pra, reprint, superscriptaddress]{revtex4-1}

\pdfoutput=1

\usepackage{header}

\graphicspath{{graphics/}}

\begin{document}

\bibliographystyle{apsrev}

\title{On the Complexity of Random Quantum Computations and the Jones Polynomial}

\author{Ryan L. Mann}
\email{mail@ryanmann.org}
\homepage{http://www.ryanmann.org}
\affiliation{Centre for Quantum Software and Information, Faculty of Engineering \& Information Technology, University of Technology Sydney, NSW 2007, Australia}

\author{Michael J. Bremner}
\affiliation{Centre for Quantum Software and Information, Faculty of Engineering \& Information Technology, University of Technology Sydney, NSW 2007, Australia}
\affiliation{Centre for Quantum Computation and Communication Technology, Faculty of Engineering \& Information Technology, University of Technology Sydney, NSW 2007, Australia}

\begin{abstract}
There is a natural relationship between Jones polynomials and quantum computation. We use this relationship to show that the complexity of evaluating relative-error approximations of Jones polynomials can be used to bound the classical complexity of approximately simulating random quantum computations. We prove that random quantum computations cannot be classically simulated up to a constant total variation distance, under the assumption that (1) the Polynomial Hierarchy does not collapse and (2) the average-case complexity of relative-error approximations of the Jones polynomial matches the worst-case complexity over a constant fraction of random links. Our results provide a straightforward relationship between the approximation of Jones polynomials and the complexity of random quantum computations.
\end{abstract}

\maketitle

\section{Introduction}
\label{section:Introduction}

The complexity of quantum computation is completely determined by the complexity of quantum circuit amplitudes. These amplitudes can encode the solution to computationally hard problems, such as Jones polynomials~\cite{aharonov2009polynomial}, Tutte polynomials~\cite{aharonov2007polynomial}, and matrix permanents~\cite{scheel2004permanents, rudolph2009simple}. Unfortunately, quantum mechanics does not provide us with a method for directly measuring these amplitudes or their corresponding probabilities. We must instead infer approximations to them via repeated computations.

There is often a significant difference between the complexity of an exact evaluation of a function and an approximation to it. For example, in the case of the ferromagnetic Ising model, an exact evaluation of its partition function is \mbox{\textsc{\#P}-hard}. However, a relative-error approximation can be achieved with a classical computer in polynomial time~\cite{jerrum1993polynomial}. Another interesting example is the Jones polynomial. Exactly computing the Jones polynomial is \mbox{\textsc{\#P}-hard}~\cite{jaeger1990computational}. However, unlike the ferromagnetic Ising model, the Jones polynomial retains this complexity for relative-error approximations~\cite{kuperberg2009hard}. It is known that, for the same class of Jones polynomials, computing additive-error approximations is \mbox{\textsc{BQP}-hard}~\cite{aharonov2011bqp}. Therefore, it seems unlikely that quantum computers can produce relative-error approximations of Jones polynomials in polynomial time.

We show that the complexity of evaluating relative-error approximations of Jones polynomials can be used to bound the classical complexity of approximately simulating random quantum computations. Under the assumption that (1) the Polynomial Hierarchy \mbox{(\textsc{PH})} does not collapse~\cite{papadimitriou2003computational} and (2) the average-case complexity of relative-error approximations of the Jones polynomial matches the worst-case complexity over a constant fraction of random links \mbox{(Conjecture~\ref{conjecture:AverageCaseComplexityJonesPolynomial})}, we prove that random quantum computations cannot be classically simulated up to a constant total variation distance \mbox{(Theorem~\ref{theorem:ClassicalSimulationRandomBraids})}. This argument follows as a natural extension to those given for Instantaneous Quantum Polynomial-time (IQP) circuits~\cite{bremner2010classical, bremner2016average} and for other classes of random quantum circuits~\cite{boixo2016characterizing}, when combined with results on approximate designs~\cite{harrow2009random, brandao2016local}. Our results provide a straightforward relationship between the approximation of Jones polynomials and the complexity of random quantum computations.

Many quantum circuit classes can be associated with functions that are \mbox{\textsc{\#P}-hard} to evaluate up to a relative error. This feature has been used to construct arguments in favour of a separation between the power of classical and quantum computation (for a review on this topic see Ref.~\cite{lund2017quantum} and Ref.~\cite{harrow2017quantum}). While we do not believe that quantum computers can exactly evaluate such functions, they play a vital role in defining the complexity of sampling from the output probability distribution of quantum circuits. Terhal and DiVincenzo~\cite{terhal2004adptive} first used this feature to bound the capability of classical computers to simulate constant-depth quantum computations. This was later extended to the problem of sampling from linear optical networks~\cite{aaronson2011computational} and IQP circuits~\cite{bremner2010classical}.

Aaronson and Arkhipov~\cite{aaronson2011computational} proved an important relationship between the complexity of approximate sampling and the average-case complexity of relative-error approximations to counting problems. They showed that the complexity of evaluating relative-error approximations to matrix permanents can be used to bound the classical complexity of sampling from random linear optical networks up to a constant total variation distance --- a notion of approximation that is realistic for quantum computation. They conjecture that (1) the average-case complexity of the permanent of Gaussian matrices is \mbox{\textsc{\#P}-hard} and (2) the permanent of Gaussian matrices satisfies a certain anti-concentration bound. Assuming that these conjectures are true, they show that the existence of an efficient classical algorithm which can approximately sample from these networks would imply the collapse of the Polynomial Hierarchy~\cite{aaronson2011computational}. A similar result was proven for IQP circuits~\cite{bremner2016average} --- extending this argument to the quantum circuit model under a different average-case complexity conjecture, where the equivalent anti-concentration conjecture could be proven.

These sampling problems are not just a good candidate for proving a separation between classical and quantum computation, but also for providing experimental benchmarks~\cite{boixo2016characterizing, harrow2017quantum}. This has motivated the study of many other sampling problems. Each of these conjecture the equivalence of the average-case and worst-case complexity of relative-error approximations of a given function. These include: (1) the permanent of Gaussian matrices~\cite{aaronson2011computational}, (2) the gap of degree-three polynomials over $\mathbb{F}_2$~\cite{bremner2016average, miller2017quantum}, (3) output probabilities of conjugated Clifford circuits~\cite{bouland2017quantum}, and (4) complex-temperature Ising model partition functions over dense~\cite{bremner2016average}, sparse~\cite{bremner2017achieving}, and three-dimensional models~\cite{boixo2016characterizing, gao2017quantum, hangleiter2017anti}.

These average-case complexity conjectures are each associated with a class of quantum circuits. These quantum circuits are not thought to be universal for quantum computation, with the exception of the three-dimensional Ising model case, but nonetheless become universal under post-selection. Understanding the distinctions between these conjectures is essential for understanding the relationship between these classes of quantum circuits. However, resolving such conjectures would require non-relativising techniques~\cite{aaronson2016complexity}. We therefore expect this to be a hard open problem.

We consider the problem of sampling from random quantum computations that are distributed according to an approximate unitary \mbox{$(t\geq2)$-design}. We observe that these approximate unitary designs produce output probability distributions that satisfy an anti-concentration bound. This bound is used to prove that if there exists an efficient classical algorithm which can sample from these distributions up to a constant total variation distance, then Stockmeyer's Counting Theorem \mbox{(Theorem~\ref{theorem:StockmeyerCountingTheorem})} can be used to produce relative-error approximations to a constant fraction of their output probabilities \mbox{(Theorem~\ref{theorem:ClassicalSimulationRandomQuantumComputations})}. This same observation has been used to establish arguments for the complexity of random quantum circuits~\cite{boixo2016characterizing, hangleiter2017anti} and conjugated Clifford circuits~\cite{bouland2017quantum}.

We define a natural model of random links via the braid group. A random braid is generated by applying generators of the braid group uniformly at random. A random link is then the plat closure of a random braid. We show that the output probability amplitudes of random quantum computations are proportional to the Jones polynomial of a random link. Furthermore, we show that in the $k^{\mathrm{th}}$ path model representation with $k=5$ or $k\geq7$, random braids on $2n$ strands of length \mbox{$\Omega[n(n+\log(1/\epsilon))]$} form an \mbox{$\epsilon$-approximate} unitary \mbox{$2$-design} \mbox{(Corollary~\ref{corollary:RandomBraidsDesign})}. This leads us to conjecture that it is \mbox{\textsc{\#P}-hard} to approximate the Jones polynomial, up to a relative error, on at least a constant fraction of random links \mbox{(Conjecture~\ref{conjecture:AverageCaseComplexityJonesPolynomial})}. This provides a natural conjecture for bounding the classical complexity of simulating random quantum computations.

This paper is structured as follows. In \mbox{Section~\ref{section:RandomQuantumComputations}}, we provide an introduction to random quantum computations and approximate unitary designs. We then state our result on the classical simulation of random quantum computations. In \mbox{Section~\ref{section:KnotsBraidsAndTheJonesPolynomial}}, we briefly introduce the theory of knots, braids, and the Jones polynomial. We review the relationship between Jones polynomials and quantum computing in \mbox{Section~\ref{section:JonesPolynomialAndQuantumComputing}}. In \mbox{Section~\ref{section:RandomQuantumComputationsAndRandomLinks}}, we relate the complexity of random quantum computations to the complexity of approximating the Jones polynomial of random links. Finally, we conclude in \mbox{Section~\ref{section:ConclusionAndOutlook}} with some remarks and open problems.

\section{Random Quantum Computations}
\label{section:RandomQuantumComputations}

A \emph{random quantum computation} is the action of (1) preparing an initial state, (2) applying a randomly chosen unitary matrix, and (3) measuring in the computational basis. This is equivalent to sampling from a probability distribution $\mathcal{D}_U$, where $U$ is a randomly chosen unitary matrix.
\begin{definition}[$\mathcal{D}_U$]
    For a $d \times d$ unitary matrix $U$, we define $\mathcal{D}_U$ to be the probability distribution over integers $x\in[d]$, given by
    \begin{align}
        \textbf{Pr}[x] := \abs{\bra{x}U\ket{0}}^2. \notag
    \end{align}
\end{definition}
It is natural to consider unitary matrices drawn from the uniform distribution. The uniform distribution over the unitary group $\textrm{U}(d)$ is defined by the \emph{Haar measure}, which is the unique translation-invariant measure on the group. Unfortunately, random unitary matrices drawn from the Haar measure cannot be implemented efficiently by a quantum computer as they typically require an exponential number of gates~\cite{knill1995approximation}.

For our purposes, it is important that the random quantum computations can be implemented efficiently. We achieve this by weakening the requirement that the unitary matrices are drawn from the Haar measure. Instead, we require only that the unitary matrices are drawn from a distribution that is close to the Haar measure.

A \emph{unitary \mbox{$t$-design}} is a distribution over a finite set of unitary matrices which imitates the properties of the Haar measure up to the $t^{\mathrm{th}}$ moment. For convenience, let $\mathrm{Hom}_{(t,t)}(\textrm{U}(d))$ be the set of polynomials homogeneous of degree $t$ in the matrix elements of $U$ and homogeneous of degree $t$ in the matrix elements of $U^*$.
\begin{definition}[Unitary \mbox{$t$-design}~\cite{roy2009unitary}]
A distribution \mbox{$\mathcal{D}=\{p_i, U_i\}$} over unitary matrices in dimension $d$ is a unitary \mbox{$t$-design} if, for any polynomial \mbox{$f \in \mathrm{Hom}_{(t,t)}(\textrm{U}(d))$},
\begin{align}
    \sum_{U_i \in \mathcal{D}}p_if(U_i) = \int\limits_{\mathclap{\textrm{U}(d)}}f(U)dU. \notag
\end{align}
\end{definition}

\begin{definition}[\mbox{$\epsilon$-approximate} unitary \mbox{$t$-design}]
A distribution \mbox{$\mathcal{D}=\{p_i, U_i\}$} over unitary matrices in dimension $d$ is an \mbox{$\epsilon$-approximate} unitary \mbox{$t$-design} if, for any polynomial \mbox{$f\in\mathrm{Hom}_{(t,t)}(\textrm{U}(d))$},
\begin{align}
    (1-\epsilon)\int\limits_{\mathclap{\textrm{U}(d)}}f(U)dU \leq \sum_{U_i \in \mathcal{D}}p_if(U_i) \leq (1+\epsilon)\int\limits_{\mathclap{\textrm{U}(d)}}f(U)dU. \notag
\end{align}
\end{definition}

Brandao, Harrow, and Horodecki~\cite{brandao2016local} showed that \mbox{$G$-local} random quantum circuits acting on $n$ qudits composed of polynomially many gates form an approximate unitary \mbox{$\mathrm{poly}(n)$-design}. Here, \mbox{$G=\{g_i\}_{i=1}^{m}$} is a universal set of gates containing inverses with each \mbox{$g_i\in\textrm{U}(d^2)$} composed of algebraic entries.
\begin{definition}[\mbox{$G$-local} random quantum circuit]
    At each time step, two indices, $i$ and $j$, are chosen uniformly at random from $[m]$ and $[n-1]$, respectively. The gate $g_i$ is then applied to the two neighbouring qudits $j$ and $j+1$.
\end{definition}
\begin{theorem}[Brandao, Harrow, and Horodecki~\cite{brandao2016local}]
    \label{theorem:RandomQuantumCircuitsPolynomialDesigns}
    Fix $d\geq2$. Let \mbox{$G=\{g_i\}_{i=1}^{m}$} be a universal set gates containing inverses with each \mbox{$g_i\in\emph{\textrm{U}}(d^2)$} composed of algebraic entries. There exists a constant \mbox{$\lambda=\lambda(G)>0$} such that \mbox{$G$-local} random quantum circuits of length
    \begin{align}
        \lambda n\left\lceil\log_{d}(4t)\right\rceil^2t^5t^{3.1/\log(d)}\left[nt\log\left(d\right)+\log(1/\epsilon)\right] \notag
    \end{align}
    form an \mbox{$\epsilon$-approximate} unitary \mbox{$t$-design}.
\end{theorem}
We shall, therefore, restrict our attention to random quantum computations where the unitary matrices are drawn from an \mbox{$\epsilon$-approximate} unitary \mbox{$(t\geq2)$-design}. We are interested in a classical simulation of random quantum computations, for which we have the following result:
\begin{restatable}{theorem}{ClassicalSimulationRandomQuantumComputations}
    \label{theorem:ClassicalSimulationRandomQuantumComputations}
    Let $U$ be a $d \times d$ unitary matrix distributed according to an \mbox{$\epsilon$-approximate} unitary \mbox{$(t\geq2)$-design} and let $\mathcal{D}_U$ be its corresponding probability distribution. Suppose that there is a classical polynomial-time algorithm $C$, which, for any $U$, samples from a probability distribution $\mathcal{D}^\prime$, such that \mbox{$\norm{\mathcal{D}^\prime-\mathcal{D}_U}_1 \leq \mu$}. Then, for any $\gamma$ such that \mbox{$0 < \gamma < 1-\epsilon$}, there is an $\emph{\textsc{FBPP}}^{\emph{\textsc{NP}}^C}$ algorithm which approximates $\abs{\bra{0}U\ket{0}}^2$ up to a relative error \mbox{$\frac{4\mu(1+\epsilon)^2}{\gamma(1-\epsilon-\gamma)^2}+o(1)$} on at least a $\frac{(1-\epsilon-\gamma)^2}{4(1+\epsilon)}$ fraction of matrices.
\end{restatable}

We prove \mbox{Theorem~\ref{theorem:ClassicalSimulationRandomQuantumComputations}} and several supporting lemmas in \mbox{Appendix~\ref{section:ProofOfTheoremCSRQS}}. \mbox{Theorem~\ref{theorem:ClassicalSimulationRandomQuantumComputations}} tells us that, if there exists an efficient classical algorithm which can approximately sample from any random quantum computation, then, there is an $\textsc{FBPP}^\textsc{NP}$ algorithm which can approximate \mbox{$\abs{\bra{0}U\ket{0}}^2$} up to a relative error for a fraction of matrices $U$. Suppose that this algorithm solves a \mbox{\textsc{\#P}-hard} problem, then, by Toda's Theorem~\cite{toda1991pp}, the Polynomial Hierarchy collapses to its third level.

\begin{theorem}[Toda~\cite{toda1991pp}]
    \label{theorem:Toda}
    \begin{align}
        {\normalfont\textsc{PH}} \subseteq {\normalfont\textsc{P}}^{\normalfont\textsc{\#P}}. \notag
    \end{align}
\end{theorem}

In \mbox{Section~\ref{section:RandomQuantumComputationsAndRandomLinks}}, we show that \mbox{$\abs{\bra{0}U\ket{0}}^2$} is proportional to the Jones polynomial of a random link, which is known to be \mbox{\textsc{\#P}-hard} to approximate up to a relative error in the worst case~\cite{kuperberg2009hard}. We conjecture that this remains true in the average case.

\section{Knots, Braids, and the Jones Polynomial}
\label{section:KnotsBraidsAndTheJonesPolynomial}

We now briefly introduce the theory of knots, braids, and the Jones polynomial.

\begin{definition}[Knot]
A knot $K$ is subset of points in $\mathbb{R}^3$ that is homeomorphic to a circle.
\end{definition}

Informally, a knot is a tangled strand of string with the open ends closed to form a loop. Much like the everyday knots that we use when we tie our shoelaces, ties, and so on --- mathematical knots are exactly that, except that the open ends are fused together.

The most simple knot you can think of is the \emph{unknot}, also called the \emph{trivial knot}, which is a closed loop without a knot \mbox{(Fig.~\ref{figure:Unknot})}. Other examples of knots include the \emph{trefoil knot} \mbox{(Fig.~\ref{figure:Trefoil})}, and the \emph{figure eight knot} \mbox{(Fig.~\ref{figure:FigureEight})}.

\begin{figure}[ht]
    \centering
    \begin{subfigure}{.30\linewidth}
        \centering
        \includegraphics[width=.8\linewidth]{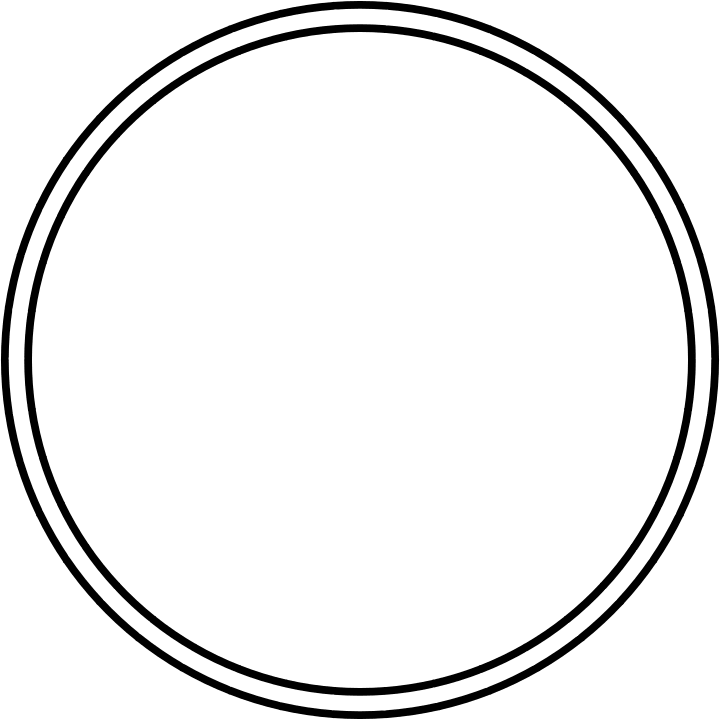}
        \caption{The unknot.}
        \label{figure:Unknot}
    \end{subfigure}
    \begin{subfigure}{.30\linewidth}
        \centering
        \includegraphics[width=.8\linewidth]{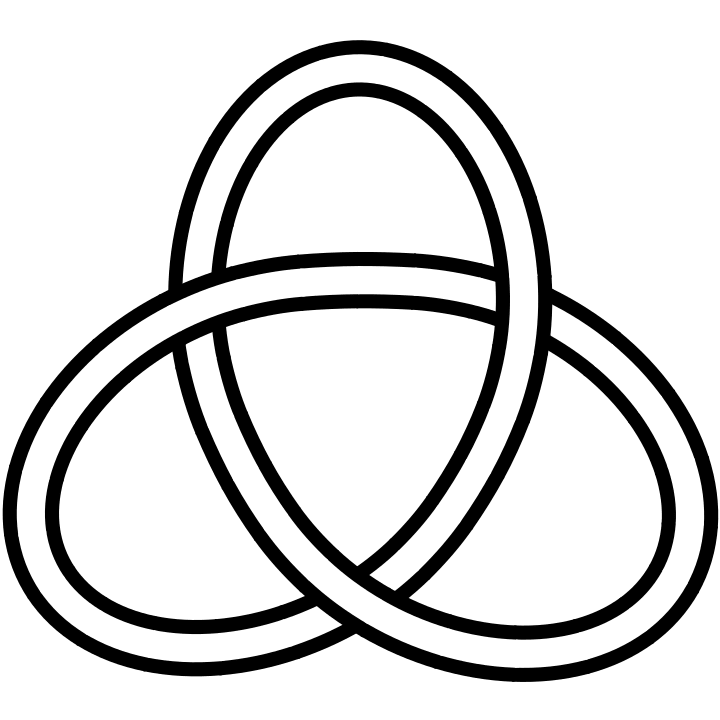}
        \caption{The trefoil knot.}
        \label{figure:Trefoil}
    \end{subfigure}
    \begin{subfigure}{.30\linewidth}
        \centering
        \includegraphics[width=.8\linewidth]{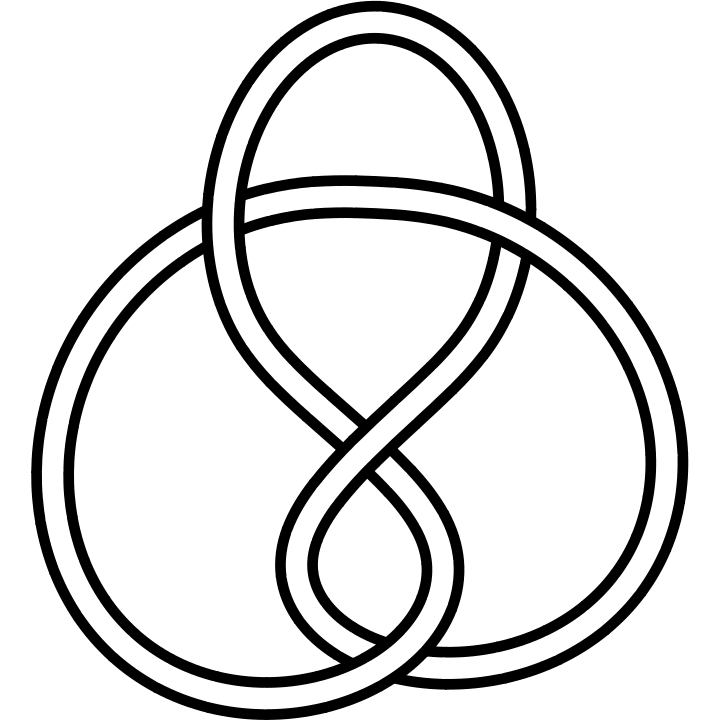}
        \caption{The figure eight knot.}
        \label{figure:FigureEight}
    \end{subfigure}
    \caption{Examples of basic knots.}
\end{figure}

We have seen how a knot is an embedding of a circle in $\mathbb{R}^3$. We can now generalise this idea by considering an embedding of multiple circles in $\mathbb{R}^3$.

\begin{definition}[Link]
A link $L$ is a finite disjoint union of knots $L=\bigcup_iK_i$. Each knot $K_i$ in the union is called a \emph{component} of the link.
\end{definition}

\begin{definition}[Oriented link]
An oriented link is a link in which each component is assigned an orientation.
\end{definition}

We can now see that a knot is a link of only one component. The generalisation of the unknot to a link on $n$ components is called the \emph{unlink}, which is a collection of $n$ unknots that are not interlinked. An example of a slightly more interesting link is the \emph{Borromean rings} link \mbox{(Fig.~\ref{figure:BorromeanRings})}, which has the property that removing any single component of the link gives the two component unlink.

\begin{figure}[ht]
    \centering
    \includegraphics[width=.30\linewidth]{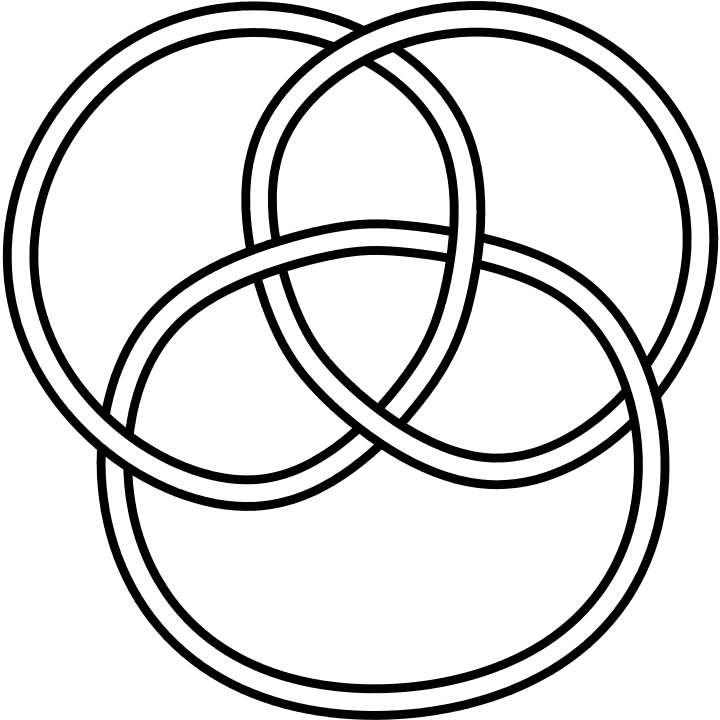}
    \caption{The Borromean rings link.}
    \label{figure:BorromeanRings}
\end{figure}

A important problem in knot theory is the \emph{link recognition problem} --- given two links are they the same? To answer this, we must first ask, what does it mean for two links to be the same?

\begin{definition}[Link equivalence]
Two links $L_1$ and $L_2$ are said to be equivalent if there exists a orientation-preserving homeomorphism \mbox{$f:\mathbb{R}^3\to\mathbb{R}^3$} so that \mbox{$f(L_1)=L_2$}.
\end{definition}

Essentially, two links are equivalent if they can be deformed into one another. We can prove that two links are equivalent by producing a set of instructions that will deform one link into the other. However, proving that two links are not equivalent is much more difficult, as we would need to prove that no set of instructions exist.

Link invariants are an important concept in knot theory as they allow us to study the link recognition problem.

\begin{definition}[Link invariant]
A link invariant is a function from the set of links to some other set, such that the output of the function depends only on the equivalence class of the link.
\end{definition}

\begin{definition}[Jones polynomial~\cite{jones1985polynomial}]
The Jones polynomial $V_L(\omega)$ is a link invariant, which assigns to each oriented link a Laurent polynomial in the variable $\omega^{1/2}$.
\end{definition}

The Jones polynomial is characterised by the \emph{skein relation} and the normalisation that the Jones polynomial of the unknot \mbox{$V_\bigcirc(\omega)=1$}.

\begin{definition}[Skein relation]
Given three links $L_-$, $L_0$, and $L_+$ that are identical, except for a local region where they differ according to \mbox{Fig.~\ref{figure:SkeinRelation}}, then the following skein relation holds
\begin{align}
    (\omega^{1/2}-\omega^{-1/2})V_{L_0}(\omega) = \omega^{-1}V_{L_+}(\omega)-\omega V_{L_-}(\omega). \notag
\end{align}
\end{definition}

\begin{figure}[ht]
    \centering
    \includegraphics[width=.70\linewidth]{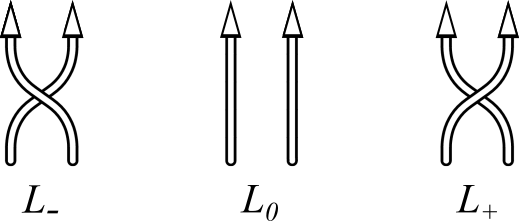}
    \caption{Diagrams for the skein relation.}
    \label{figure:SkeinRelation}
\end{figure}

The skein relation is sufficient for a recursive computation of the Jones polynomial of a link. It follows that the Jones polynomial of a link can be computed in time exponential in the number of crossings. A classic result of Jaeger, Vertigan, and Welsh~\cite{jaeger1990computational} states that exactly computing the Jones polynomial $V_L(\omega)$ of a link is \mbox{\textsc{\#P}-hard} except when $\omega$ is one of a few special points. Bordewich et al.~\cite{bordewich2005approximate} showed that it is \mbox{\textsc{BQP}-hard} to approximate the Jones polynomial up to an additive error. Kuperberg~\cite{kuperberg2009hard} proved that it remains \mbox{\textsc{\#P}-hard} to approximate the Jones polynomial up to a relative error.

\begin{theorem}[Jaeger, Vertigan, and Welsh~\cite{jaeger1990computational}]
    Evaluating the Jones polynomial $V_L(\omega)$ of a link is \mbox{\emph{\textsc{\#P}-hard}} except when \mbox{$\omega=\pm\exp(2\pi i/k)$} with \mbox{$k\in\{1,2,3,4,6\}$} when it can be evaluated in polynomial time.
\end{theorem}

We now introduce the theory of braids, which provides us with a convenient way to represent any link.
\begin{definition}[Braid]
Let
\begin{align}
    A = \{(x, 0, 0) \mid x \in \mathbb{Z}^+, x \leq n \,\}, \notag \\
    B = \{\,(x, 0, 1) \mid x \in \mathbb{Z}^+, x \leq n \,\}. \notag
\end{align}
Then, an \mbox{$n$-strand} braid is a collection of non-intersecting smooth paths in $\mathbb{R}^3$ connecting the points in $A$ to the points in $B$.
\end{definition}
Informally, a braid is a collection of strands of string that  may cross over and under each other, and must always move from left to right. An example of a braid is given in \mbox{Fig.~\ref{figure:Braid}}.
\begin{figure}[ht]
    \centering
    \includegraphics[width=.70\linewidth]{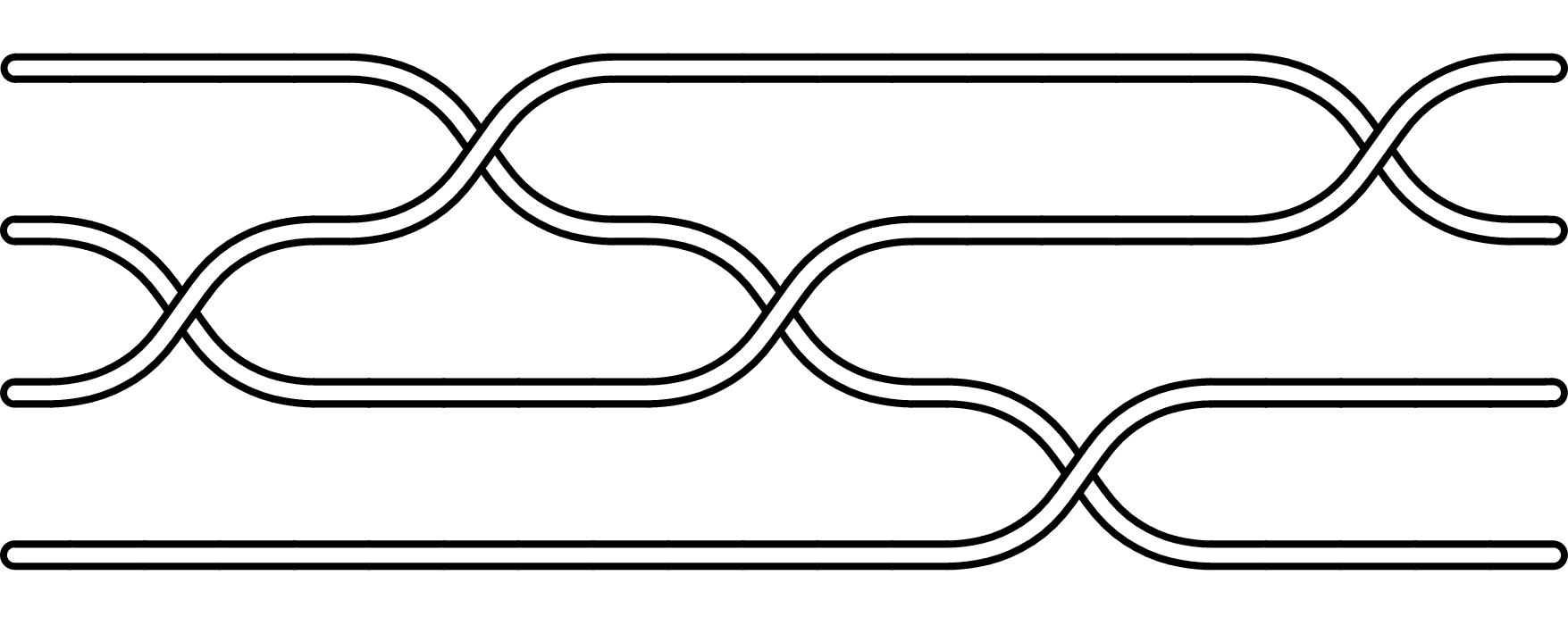}
    \caption{An example of a braid on $4$ strands.}
    \label{figure:Braid}
\end{figure}

The set of all braids on $n$ strands form an infinite group $B_n$, generated by the $n-1$ generators $\{\sigma_i\}$ and their inverses $\{\sigma_i^{-1}\}$. The generator $\sigma_i$ crosses the $i^{\mathrm{th}}$ strand over the $(i+1)^{\mathrm{th}}$ strand and its inverse $\sigma_i^{-1}$ crosses the $i^{\mathrm{th}}$ strand under the $(i+1)^{\mathrm{th}}$ strand.
\begin{definition}[Braid group]
    The braid group on $n$ strands $B_n$ is the group given by the Artin presentation
    \begin{align}
        \left\langle\{\sigma_i\}_{i=1}^n\middle|
        \begin{aligned}
        \sigma_i\sigma_{i+1}\sigma_i =& \sigma_{i+1}\sigma_i\sigma_{i+1} && \text{for } 1\leq i \leq n-2 \\
        \sigma_i\sigma_j =& \sigma_j\sigma_i && \text{for }\;\abs{i-j} \geq 2
        \end{aligned}\;\right\rangle. \notag
    \end{align}
\end{definition}
Each braid can be described by a \emph{braid word}.
\begin{definition}[Braid word]
    A braid word is word on the set of generators $\{\sigma_i\}$ and their inverses $\{\sigma_i^{-1}\}$. The \emph{length} of a braid word is the number of characters in the word.
\end{definition}

We can connect the endpoints of any braid in a number of ways to form a link. For a braid with an even number of strands a natural way to do this is by the \emph{plat closure}.
\begin{definition}[Plat closure]
    The plat closure of a \mbox{$2n$-strand} braid \mbox{$b \in B_{2n}$} is the link formed by connecting pairs of adjacent strands on the left and the right of the braid. The link that is formed by the plat closure of the braid is often denoted $b^{pl}$.
\end{definition}
Alexander~\cite{alexander1923lemma} showed that we can generate all possible links this way. We can, therefore, describe any link as the closure of a braid given by its braid word.
\begin{theorem}[Alexander~\cite{alexander1923lemma}]
    Every link can be represented by the closure of some braid.
\end{theorem}

\section{The Jones Polynomial and Quantum Computing}
\label{section:JonesPolynomialAndQuantumComputing}

Freedman, Kitaev, and Wang~\cite{freedman2002simulation} established a quantum algorithm for additively approximating the Jones polynomial at any principle root of unity in polynomial time. This algorithm was later formalised by Aharonov, Jones, and Landau~\cite{aharonov2009polynomial}. Freedman, Larsen, and Wang~\cite{freedman2002modular} proved that when \mbox{$\omega=\exp(2\pi i/k)$} is a \emph{principle non-lattice root of unity}, i.e. $k=5$ or $k\geq7$, the problem of additively approximating the Jones polynomial is universal for quantum computation. Aharonov and Arad~\cite{aharonov2011bqp} extended this result to values of $k$ that grow polynomially with the number of strands and crossings.
\begin{theorem}[Aharonov and Arad~\cite{aharonov2011bqp}]
    Let $\omega$ be a principle non-lattice root of unity, and let \mbox{$b \in B_{2n}$} be a braid. Then, the problem of additively approximating the Jones polynomial $V_{b^{pl}}(\omega)$ to within the same accuracy as the Aharonov-Jones-Landau algorithm~\emph{\cite{aharonov2009polynomial}} is \mbox{\emph{\textsc{BQP}-hard}}.
\end{theorem}

The Aharonov-Jones-Landau algorithm is based on the \emph{path model representation of the braid group}~\cite{jones1983braid, jones1985polynomial}, which is unitary when \mbox{$\omega=\exp(2\pi i/k)$} is a principle root of unity. For an integer $k$, the $k^\mathrm{th}$ path model representation of the braid group $B_{2n}$ is defined on the vector space spanned by walks of length $2n$, on a $k-1$ vertex path graph $G_k$, which start and finish on the first vertex.

To calculate the dimension of this vector space it is sufficient to count the number of walks of length $2n$ on the graph $G_k$. From a combinatorial perspective, the walks on the graph $G_k$ can be seen as \emph{Dyck paths} of length $2n$, which never go above a height $k-2$. It is well known that the number of Dyck paths of length $2n$ is the $n^{\mathrm{th}}$ \emph{Catalan number}, which provides an upperbound for the dimension of the vector space.
\begin{definition}[Catalan number]
    The $n^{\mathrm{th}}$ Catalan number is defined by
    \begin{align}
        C_n := \frac{1}{(n+1)}\binom{2n}{n}. \notag
    \end{align}
\end{definition}
\begin{claim}
    \label{claim:CatalanNumberUpperbound}
    For $n\geq1$,
    \begin{align}
        C_n < 4^n. \notag
    \end{align}
\end{claim}
\begin{proof}
    The claim follows directly from Stirling's approximation for factorials.
\end{proof}

In this representation, each braid \mbox{$b \in B_{2n}$} is mapped to a unitary matrix $\rho_k(b)$ composed of algebraic entries. These unitary matrices have the property that the expectation value \mbox{$\bra{0}\rho_k(b)\ket{0}$} is proportional, up to an efficiently computable factor, to the Jones polynomial $V_{b^{pl}}(\omega)$ of the plat closure of $b$. Aharonov, Jones, and Landau~\cite{aharonov2009polynomial} showed that such representations can be implemented efficiently on a quantum computer.

In their construction, the unitary representation of each generator $\rho_k(\sigma_i^\pm)$ of the braid group $B_{2n}$ acts on a subspace of the Hilbert space of qudits. The Solovay-Kitaev theorem~\cite{kitaev2002classical} guarantees that these unitary matrices can be implemented efficiently. An entire braid \mbox{$b \in B_{2n}$} is implemented efficiently by applying the corresponding unitary matrix of each generator in the order of the braid word of $b$.

\section{Random Quantum Computations and Random Links}
\label{section:RandomQuantumComputationsAndRandomLinks}

We now relate random quantum computations and the Jones polynomial of random links. We define a \emph{random link} to be the plat closure of a \emph{random braid}.
\begin{definition}[Random braid]
    A random braid on $2n$ strands is generated by uniformly at random choosing generators from the set $\{\sigma_i^\pm\}_{i=1}^{2n-1}$.
\end{definition}
\begin{definition}[Random link]
    A random link is generated by the plat closure of a random braid.
\end{definition}

In the $k^{\mathrm{th}}$ path model representation the generators of the braid group $\{\sigma_i^\pm\}$ are mapped to unitary matrices $\{\rho_k(\sigma_i^\pm)\}$. In this representation, a random braid is equivalent to a product of random matrices chosen uniformly at random from the set $\{\rho_k(\sigma_i^\pm)\}$. Since each $\rho_k(\sigma_i^\pm)$ acts on a subspace of the Hilbert space of qudits, a random braid is equivalent to a \mbox{$G$-local} random quantum circuit, with the number of strands proportional to the number of qudits. When $k=5$ or $k\geq7$ these gates are universal for quantum computation.

\begin{theorem}
    \label{theorem:BraidRepresentationTDesign}
    In the $k^{\mathrm{th}}$ path model representation with $k=5$ or $k\geq7$, there exists a constant $\lambda>0$, such that random braids on $2n$ strands of length
    \begin{align}
        \lambda n\left\lceil\log_2(4t)\right\rceil^2t^5t^{3.1/\log(2)}\left[t\log\left(C_n\right)+\log(1/\epsilon)\right], \notag
    \end{align}
    form an \mbox{$\epsilon$-approximate} unitary \mbox{$t$-design}.
\end{theorem}
\begin{proof}
    The proof follows from combining \mbox{Theorem~\ref{theorem:RandomQuantumCircuitsPolynomialDesigns}} with the fact that the dimension of the vector space in the path model representation is bounded from above by the $n^{\mathrm{th}}$ Catalan number and that the local dimension is bounded from below by $2$.
\end{proof}

\begin{corollary}
    \label{corollary:RandomBraidsDesign}
    In the $k^{\mathrm{th}}$ path model representation with $k=5$ or $k\geq7$, there exists a constant $\lambda>0$, such that random braids on $2n$ strands of length
    \begin{align}
        \lambda n\left[n+\log(1/\epsilon)\right], \notag
    \end{align}
    form an \mbox{$\epsilon$-approximate} unitary \mbox{$2$-design}.
\end{corollary}
\begin{proof}
    The proof follows from setting $t=2$ in \mbox{Theorem~\ref{theorem:BraidRepresentationTDesign}} and from the upperbound for the $n^{\mathrm{th}}$ Catalan number found in \mbox{Claim~\ref{claim:CatalanNumberUpperbound}}.
\end{proof}

We now relate the classical simulation of random quantum computations and the complexity of approximating the Jones polynomial of random links.
\begin{theorem}
    \label{theorem:ClassicalSimulationRandomBraids}
    Fix \mbox{$0<\epsilon<1$}. Let $k=5$ or $k\geq7$ be an integer, and \mbox{$\omega=\exp(2\pi i/k)$} its corresponding root of unity. Let \mbox{$b \in B_{2n}$} be a random braid on $2n$ strands of length \mbox{$\Omega\left[n(n+\log(1/\epsilon))\right]$}. Let $\rho_k(b)$ be the $k^{\mathrm{th}}$ path model representation of $b$, and let $\mathcal{D}_{\rho_k(b)}$ be its corresponding probability distribution. Suppose that there is a classical polynomial-time algorithm $C$, which, for any $b$, samples from a probability distribution $\mathcal{D}^\prime$, such that \mbox{$\norm{\mathcal{D}^\prime-\mathcal{D}_{\rho(b)}}_1 \leq \mu$} and assume that \mbox{Conjecture~\ref{conjecture:AverageCaseComplexityJonesPolynomial}} holds. Then, there is a $\emph{\textsc{BPP}}^\emph{\textsc{NP}}$ algorithm for solving any problem in $\emph{\textsc{P}}^\emph{\textsc{\#P}}$ and by Toda's Theorem the Polynomial Hierarchy collapses to its third level.
\end{theorem}
\begin{proof}
    The proof follows from combining \mbox{Theorem~\ref{theorem:ClassicalSimulationRandomQuantumComputations}}, \mbox{Corollary~\ref{corollary:RandomBraidsDesign}}, and Toda's Theorem \mbox{(Theorem~\ref{theorem:Toda})}.
\end{proof}

\begin{conjecture}
    \label{conjecture:AverageCaseComplexityJonesPolynomial}
    In the notation of \mbox{Theorem~\ref{theorem:ClassicalSimulationRandomBraids}}. For some \mbox{$0<\gamma<1-\epsilon$}, it is \mbox{\textsc{\#P}-hard} to approximate the Jones polynomial $V_{b^{pl}}(\omega)$ up to a relative error \mbox{$\frac{4\mu(1+\epsilon)^2}{\gamma(1-\epsilon-\gamma)^2}+o(1)$} on at least a $\frac{(1-\epsilon-\gamma)^2}{4(1+\epsilon)}$ fraction of random braids.
\end{conjecture}

\mbox{Conjecture~\ref{conjecture:AverageCaseComplexityJonesPolynomial}} is based on the average-case complexity of relative-error approximations of Jones polynomials. It is known that it is \mbox{\textsc{\#P}-hard} to approximate the Jones polynomial up to a relative error in the worst case~\cite{kuperberg2009hard}. Therefore, \mbox{Conjecture~\ref{conjecture:AverageCaseComplexityJonesPolynomial}} states that this worst-case hardness result can be extended to an average-case hardness result.

Assuming that \mbox{Conjecture~\ref{conjecture:AverageCaseComplexityJonesPolynomial}} holds and the Polynomial Hierarchy does not collapse, \mbox{Theorem~\ref{theorem:ClassicalSimulationRandomBraids}} tells us that there is no efficient classical algorithm which can sample from any random quantum computation. This implies that random quantum computations can not be efficiently simulated by a classical computer.

It is worth noting that the $5^{\mathrm{th}}$ path model representation is equivalent to the \emph{Fibonacci representation of the braid group}~\cite{shor2008estimating}. Therefore, our results extend to the random braiding of Fibonacci anyons.

\section{Conclusion \& Outlook}
\label{section:ConclusionAndOutlook}

We have provided strong evidence that simulating random quantum computations is intractable for classical computers. Specifically, we have shown that if \mbox{Conjecture~\ref{conjecture:AverageCaseComplexityJonesPolynomial}} holds and the Polynomial Hierarchy does not collapse, then there is no efficient classical algorithm which can approximately sample from the output probability distribution of random quantum computations.

There are a number of natural problems that remain to be solved. The most obvious of which is to resolve \mbox{Conjecture~\ref{conjecture:AverageCaseComplexityJonesPolynomial}}. Unfortunately, we are unaware of any proof techniques which are capable of extending the worst-case hardness result to an average-case hardness result. Moreover, the results of Aaronson and Chen~\cite{aaronson2016complexity} imply that any proof of this conjecture would require non-relativising techniques.

Another natural problem is whether \mbox{Corollary~\ref{corollary:RandomBraidsDesign}} can be strengthened to random braids of a shorter length. In \mbox{Theorem~\ref{theorem:ClassicalSimulationRandomBraids}}, the length of a random braid is determined by the requirement that in the path model representation it is distributed according to an \mbox{$\epsilon$-approximate} unitary \mbox{$2$-design}. Therefore, any improvement to this bound yields a stronger version of \mbox{Theorem~\ref{theorem:ClassicalSimulationRandomBraids}}. It is an open problem whether this bound can be improved.

It would also be interesting to adapt our results to other functions, such as Tutte polynomials~\cite{aharonov2007polynomial}, Turaev-Viro invariants~\cite{alagic2010estimating}, and matrix permanents~\cite{rudolph2009simple}. These functions are all known to be \mbox{\textsc{\#P}-hard} to compute in the worst case and \mbox{\textsc{BQP}-hard} to approximate up to an additive error.

\section*{Acknowledgements}

We thank Scott Aaronson, Sergio Boixo, Adam Bouland, Gavin Brennen, Aram Harrow, Saeed Mehraban, Ashley Montanaro, Peter Rohde, and Marco Tomamichel for helpful discussions. MJB acknowledges support from the Australian Research Council via the Future Fellowship scheme \mbox{(grant FT110101044)} and as a member of the ARC Centre of Excellence for Quantum Computation and Communication Technology (CQC2T), project number CE170100012.

\appendix

\section{Proof of Theorem~\ref*{theorem:ClassicalSimulationRandomQuantumComputations}}
\label{section:ProofOfTheoremCSRQS}

We now prove \mbox{Theorem~\ref{theorem:ClassicalSimulationRandomQuantumComputations}}, which is restated below for convenience. Our proof requires several lemmas which we prove in the remainder of the section.

\ClassicalSimulationRandomQuantumComputations*

\begin{proof}
    \mbox{Lemma~\ref{lemma:ClassicalAdditiveApproximation}} tells us that, for any \mbox{$0<\delta<1$}, there is an $\textsc{FBPP}^{\textsc{NP}^C}$ algorithm, which approximates $\abs{\bra{x}U\ket{0}}^2$, up to an additive error
    \begin{align}
       O\left[(1+o(1))\frac{\mu(1+\epsilon)}{\delta d}+\frac{\abs{\bra{x}U\ket{0}}^2}{\mathrm{poly}(n)}\right], \notag
    \end{align}
    with probability at least $1-\delta$ over the choice of $U$. Combining this with \mbox{Lemma~\ref{lemma:ApproximateDesignAnticoncentration}} and setting \mbox{$\delta=\frac{(1-\epsilon-\gamma)^2}{4(1+\epsilon)}$}, it follows that there is an $\textsc{FBPP}^{\textsc{NP}^C}$ algorithm, which approximates $\abs{\bra{0}U\ket{0}}^2$ up to a relative error \mbox{$\frac{4\mu(1+\epsilon)^2}{\gamma(1-\epsilon-\gamma)^2}+o(1)$} on at least a $\frac{(1-\epsilon-\gamma)^2}{4(1+\epsilon)}$ fraction of matrices $U$.
\end{proof}

We now prove \mbox{Lemma~\ref{lemma:ClassicalAdditiveApproximation}}, which relates the simulation of random quantum computations to approximating individual output probabilities. Our proof closely follows that of Lemma 4 from Ref.~\cite{bremner2016average}.
\begin{lemma}
    \label{lemma:ClassicalAdditiveApproximation}
    Let $U$ be a $d \times d$ unitary matrix distributed according to an \mbox{$\epsilon$-approximate} unitary \mbox{$(t\geq1)$-design} and let $\mathcal{D}_U$ be its corresponding probability distribution. Suppose that there is a classical polynomial-time algorithm $C$, which, for any $U$, samples from a probability distribution $\mathcal{D}^\prime$, such that \mbox{$\norm{\mathcal{D}^\prime-\mathcal{D}_U}_1 \leq \mu$}. Then, for any $\delta$ such that \mbox{$0<\delta<1$}, there is an $\emph{\textsc{FBPP}}^{\emph{\textsc{NP}}^C}$ algorithm, which approximates $\abs{\bra{0}U\ket{0}}^2$, up to an additive error
    \begin{align}
       O\left[(1+o(1))\frac{\mu(1+\epsilon)}{\delta d}+\frac{\abs{\bra{0}U\ket{0}}^2}{\emph{poly}(n)}\right], \notag
    \end{align}
    with probability at least $1-\delta$ over the choice of $U$.
\end{lemma}
\begin{proof}
    Define
    \begin{align}
    Q_U := \abs{\bra{0}U\ket{0}}^2, \quad T_U := \text{Pr}[\text{$C$ outputs $0$ on input $U$}]. \notag
    \end{align}
    For any $U$, we can use Stockmeyer's Counting Theorem \mbox{(Theorem~\ref{theorem:StockmeyerCountingTheorem})} to obtain a relative-error approximation to $T_U$ in $\textsc{FBPP}^{\textsc{NP}^C}$,
    \begin{align}
        \abs{T_U-T_U^\prime} \leq \frac{T_U}{\text{poly}(n)}. \notag
    \end{align}
    Then,
    \begin{align}
        \abs{Q_U-T_U\prime} \leq& \abs{Q_U-T_U}+\abs{T_U-T_U^\prime} \notag \\
        \leq& \abs{Q_U-T_U}+\frac{T_U}{\text{poly}(n)} \notag \\
        \leq& \abs{Q_U-T_U}+\frac{(Q_U+|Q_U-T_U|)}{\text{poly}(n)} \notag \\
        =& \abs{Q_U-T_U}\left(1+\frac{1}{\text{poly}(n)}\right)+\frac{Q_U}{\text{poly}(n)}. \notag
    \end{align}
    As $C$ approximates $\mathcal{D}_U$ up to an $l_1$ error $\mu$, it follows from Markov's inequality and the approximate design condition \mbox{(Lemma~\ref{lemma:ApproximateDesignCondition})} that, for any $0<\delta<1$,
    \begin{align}
        \underset{U}{\textbf{Pr}}\left[\abs{Q_U-T_U} \geq \frac{\mu(1+\epsilon)}{\delta d}\right] \leq \delta. \notag
    \end{align}
    Therefore,
    \begin{align}
       \abs{Q_U-T_U^\prime} \leq \frac{\mu(1+\epsilon)}{\delta d}\left(1+\frac{1}{\text{poly}(n)}\right)+\frac{Q_U}{\text{poly}(n)}, \notag
    \end{align}
    with probability at least $1-\delta$ over the choice of $U$.
\end{proof}

The proof of \mbox{Lemma~\ref{lemma:ClassicalAdditiveApproximation}} requires a classic result of Stockmeyer~\cite{stockmeyer1985approximation}, which allows us to approximately count in the Polynomial Hierarchy.
\begin{theorem}[Stockmeyer's Counting Theorem~\cite{stockmeyer1985approximation}]
    \label{theorem:StockmeyerCountingTheorem}
    Let \mbox{$f:\{0,1\}^n\to\{0,1\}$} be a function, and let \mbox{$F = \sum_{x\in\{0,1\}^n}f(x)$}. Then there exists an $\emph{\textsc{FBPP}}^{\emph{\textsc{NP}}^f}$ algorithm, which outputs a value $\alpha$, such that
    \begin{align}
        \abs{\alpha-F}<\Omega\left[\frac{F}{\emph{poly}(n)}\right]. \notag
    \end{align}
\end{theorem}

We now prove that unitary matrices distributed according to an \mbox{$\epsilon$-approximate} unitary \mbox{$(t\geq2)$-design} satisfy the anti-concentration bounds claimed in \mbox{Theorem~\ref{theorem:ClassicalSimulationRandomQuantumComputations}}. This was proven independently by Hangleiter et al.~\cite{hangleiter2017anti}.
\begin{lemma}
    \label{lemma:ApproximateDesignAnticoncentration}
    Let $U$ be a $d \times d$ unitary matrix distributed according to an \mbox{$\epsilon$-approximate} unitary \mbox{$(t\geq2)$-design}, then, for any unit vectors $\ket{\alpha}$, $\ket{\beta}$ and a constant \mbox{$0 \leq \gamma \leq 1-\epsilon$}, the following holds
    \begin{align}
        \underset{U}{\emph{\textbf{Pr}}}\left[\abs{\bra{\alpha}U\ket{\beta}}^2>\frac{\gamma}{d}\right] \geq \frac{(1-\epsilon-\gamma)^2}{2(1+\epsilon)}. \notag
    \end{align}
\end{lemma}
\begin{proof}
    The Paley-Zygmund inequality \mbox{(Lemma~\ref{lemma:PaleyZygmundInequality})} tells us that
    \begin{align}
        \underset{Z}{\textbf{Pr}}\left[Z>\frac{\gamma}{d}\right] \geq \left(1-\frac{\gamma}{d\mathbb{E}[Z]}\right)^2\frac{\mathbb{E}[Z]^2}{\mathbb{E}\left[Z^2\right]}, \notag
    \end{align}
    for any \mbox{$0 \leq \gamma \leq d\mathbb{E}[Z]$}. Setting \mbox{$Z=\abs{\bra{\alpha}U\ket{\beta}}^2$}, it follows from the approximate design condition \mbox{(Lemma~\ref{lemma:ApproximateDesignCondition})}, that
    \begin{align}
        \underset{U}{\textbf{Pr}}\left[Z>\frac{\gamma}{d}\right] \geq& \frac{1}{2}\left(1-\frac{\gamma}{(1-\epsilon)}\right)^2\frac{(1-\epsilon)^2}{(1+\epsilon)}\frac{(d+1)}{d} \notag \\
        \geq& \frac{1}{2}\left(1-\frac{\gamma}{1-\epsilon}\right)^2\frac{(1-\epsilon)^2}{1+\epsilon} \notag \\
        =& \frac{(1-\epsilon-\gamma)^2}{2(1+\epsilon)}, \notag
    \end{align}
    for any \mbox{$0 \leq \gamma \leq 1-\epsilon$}.
\end{proof}

The proof of \mbox{Lemma~\ref{lemma:ApproximateDesignAnticoncentration}} combines the Paley-Zygmund inequality and the approximate design condition. The Paley-Zygmund inequality bounds the probability that a non-negative random variable is small in terms of its first and second moment.
\begin{lemma}[Paley-Zygmund inequality]
    \label{lemma:PaleyZygmundInequality}
    If $Z\geq0$ is a random variable with finite variance, and if \mbox{$0 \leq \theta \leq 1$}, then
    \begin{align}
       \underset{Z}{\emph{\textbf{Pr}}}[Z>\theta\mathbb{E}[Z]] \geq (1-\theta)^2\frac{\mathbb{E}[Z]^2}{\mathbb{E}[Z^2]}. \notag
    \end{align}
\end{lemma}

We are interested in bounding the probability that the random variable \mbox{$Z=\abs{\bra{\alpha}U\ket{\beta}}^2$} is small. In the case of an exact unitary \mbox{$(t\geq2)$-design} the first and second moments of $Z$ match those of the Haar measure. For an \mbox{$\epsilon$-approximate} \mbox{$(t\geq2)$-design} the approximate design condition bounds the distance of the first and second moments of $Z$ from those of the Haar measure.
\begin{lemma}[Approximate design condition~\cite{brandao2013exponential}]
    \label{lemma:ApproximateDesignCondition}
    If $U$ is a $d \times d$ unitary matrix distributed according to an \mbox{$\epsilon$-approximate} unitary \mbox{$t$-design}, then, for any unit vectors $\ket{\alpha}$, $\ket{\beta}$ and an integer $k \leq t$,
    \begin{align}
        \frac{(1-\epsilon)}{\binom{k+d-1}{d-1}} \leq \mathbb{E}\left[\abs{\bra{\alpha}U\ket{\beta}}^{2k}\right] \leq \frac{(1+\epsilon)}{\binom{k+d-1}{d-1}}. \notag
    \end{align}
\end{lemma}

\bibliography{bibliography}

\begin{thebibliography}{38}
\expandafter\ifx\csname natexlab\endcsname\relax\def\natexlab#1{#1}\fi
\expandafter\ifx\csname bibnamefont\endcsname\relax
  \def\bibnamefont#1{#1}\fi
\expandafter\ifx\csname bibfnamefont\endcsname\relax
  \def\bibfnamefont#1{#1}\fi
\expandafter\ifx\csname citenamefont\endcsname\relax
  \def\citenamefont#1{#1}\fi
\expandafter\ifx\csname url\endcsname\relax
  \def\url#1{\texttt{#1}}\fi
\expandafter\ifx\csname urlprefix\endcsname\relax\def\urlprefix{URL }\fi
\providecommand{\bibinfo}[2]{#2}
\providecommand{\eprint}[2][]{\url{#2}}

\bibitem[{\citenamefont{Aharonov et~al.}(2009)\citenamefont{Aharonov, Jones,
  and Landau}}]{aharonov2009polynomial}
\bibinfo{author}{\bibfnamefont{D.}~\bibnamefont{Aharonov}},
  \bibinfo{author}{\bibfnamefont{V.}~\bibnamefont{Jones}}, \bibnamefont{and}
  \bibinfo{author}{\bibfnamefont{Z.}~\bibnamefont{Landau}},
  \bibinfo{journal}{Algorithmica} \textbf{\bibinfo{volume}{55}},
  \bibinfo{pages}{395} (\bibinfo{year}{2009}).

\bibitem[{\citenamefont{Aharonov et~al.}(2007)\citenamefont{Aharonov, Arad,
  Eban, and Landau}}]{aharonov2007polynomial}
\bibinfo{author}{\bibfnamefont{D.}~\bibnamefont{Aharonov}},
  \bibinfo{author}{\bibfnamefont{I.}~\bibnamefont{Arad}},
  \bibinfo{author}{\bibfnamefont{E.}~\bibnamefont{Eban}}, \bibnamefont{and}
  \bibinfo{author}{\bibfnamefont{Z.}~\bibnamefont{Landau}},
  \bibinfo{journal}{arXiv:quant-ph/0702008}  (\bibinfo{year}{2007}).

\bibitem[{\citenamefont{Scheel}(2004)}]{scheel2004permanents}
\bibinfo{author}{\bibfnamefont{S.}~\bibnamefont{Scheel}},
  \bibinfo{journal}{arXiv:quant-ph/0406127}  (\bibinfo{year}{2004}).

\bibitem[{\citenamefont{Rudolph}(2009)}]{rudolph2009simple}
\bibinfo{author}{\bibfnamefont{T.}~\bibnamefont{Rudolph}},
  \bibinfo{journal}{Physical Review A} \textbf{\bibinfo{volume}{80}},
  \bibinfo{pages}{054302} (\bibinfo{year}{2009}).

\bibitem[{\citenamefont{Jerrum and Sinclair}(1993)}]{jerrum1993polynomial}
\bibinfo{author}{\bibfnamefont{M.}~\bibnamefont{Jerrum}} \bibnamefont{and}
  \bibinfo{author}{\bibfnamefont{A.}~\bibnamefont{Sinclair}},
  \bibinfo{journal}{SIAM Journal on computing} \textbf{\bibinfo{volume}{22}},
  \bibinfo{pages}{1087} (\bibinfo{year}{1993}).

\bibitem[{\citenamefont{Jaeger et~al.}(1990)\citenamefont{Jaeger, Vertigan, and
  Welsh}}]{jaeger1990computational}
\bibinfo{author}{\bibfnamefont{F.}~\bibnamefont{Jaeger}},
  \bibinfo{author}{\bibfnamefont{D.~L.} \bibnamefont{Vertigan}},
  \bibnamefont{and} \bibinfo{author}{\bibfnamefont{D.~J.} \bibnamefont{Welsh}},
  in \emph{\bibinfo{booktitle}{Mathematical Proceedings of the Cambridge
  Philosophical Society}} (\bibinfo{organization}{Cambridge Univ Press},
  \bibinfo{year}{1990}), vol. \bibinfo{volume}{108}, pp.
  \bibinfo{pages}{35--53}.

\bibitem[{\citenamefont{Kuperberg}(2009)}]{kuperberg2009hard}
\bibinfo{author}{\bibfnamefont{G.}~\bibnamefont{Kuperberg}},
  \bibinfo{journal}{arXiv:0908.0512}  (\bibinfo{year}{2009}).

\bibitem[{\citenamefont{Aharonov and Arad}(2011)}]{aharonov2011bqp}
\bibinfo{author}{\bibfnamefont{D.}~\bibnamefont{Aharonov}} \bibnamefont{and}
  \bibinfo{author}{\bibfnamefont{I.}~\bibnamefont{Arad}}, \bibinfo{journal}{New
  Journal of Physics} \textbf{\bibinfo{volume}{13}}, \bibinfo{pages}{035019}
  (\bibinfo{year}{2011}).

\bibitem[{\citenamefont{Papadimitriou}(2003)}]{papadimitriou2003computational}
\bibinfo{author}{\bibfnamefont{C.~H.} \bibnamefont{Papadimitriou}},
  \emph{\bibinfo{title}{Computational complexity}} (\bibinfo{publisher}{John
  Wiley and Sons Ltd.}, \bibinfo{year}{2003}).

\bibitem[{\citenamefont{Bremner et~al.}(2010)\citenamefont{Bremner, Jozsa, and
  Shepherd}}]{bremner2010classical}
\bibinfo{author}{\bibfnamefont{M.~J.} \bibnamefont{Bremner}},
  \bibinfo{author}{\bibfnamefont{R.}~\bibnamefont{Jozsa}}, \bibnamefont{and}
  \bibinfo{author}{\bibfnamefont{D.~J.} \bibnamefont{Shepherd}}, in
  \emph{\bibinfo{booktitle}{Proceedings of the Royal Society of London A:
  Mathematical, Physical and Engineering Sciences}} (\bibinfo{organization}{The
  Royal Society}, \bibinfo{year}{2010}), p. \bibinfo{pages}{rspa20100301}.

\bibitem[{\citenamefont{Bremner et~al.}(2016)\citenamefont{Bremner, Montanaro,
  and Shepherd}}]{bremner2016average}
\bibinfo{author}{\bibfnamefont{M.~J.} \bibnamefont{Bremner}},
  \bibinfo{author}{\bibfnamefont{A.}~\bibnamefont{Montanaro}},
  \bibnamefont{and} \bibinfo{author}{\bibfnamefont{D.~J.}
  \bibnamefont{Shepherd}}, \bibinfo{journal}{Physical Review Letters}
  \textbf{\bibinfo{volume}{117}}, \bibinfo{pages}{080501}
  (\bibinfo{year}{2016}).

\bibitem[{\citenamefont{Boixo et~al.}(2016)\citenamefont{Boixo, Isakov,
  Smelyanskiy, Babbush, Ding, Jiang, Martinis, and
  Neven}}]{boixo2016characterizing}
\bibinfo{author}{\bibfnamefont{S.}~\bibnamefont{Boixo}},
  \bibinfo{author}{\bibfnamefont{S.~V.} \bibnamefont{Isakov}},
  \bibinfo{author}{\bibfnamefont{V.~N.} \bibnamefont{Smelyanskiy}},
  \bibinfo{author}{\bibfnamefont{R.}~\bibnamefont{Babbush}},
  \bibinfo{author}{\bibfnamefont{N.}~\bibnamefont{Ding}},
  \bibinfo{author}{\bibfnamefont{Z.}~\bibnamefont{Jiang}},
  \bibinfo{author}{\bibfnamefont{J.~M.} \bibnamefont{Martinis}},
  \bibnamefont{and} \bibinfo{author}{\bibfnamefont{H.}~\bibnamefont{Neven}},
  \bibinfo{journal}{arXiv:1608.00263}  (\bibinfo{year}{2016}).

\bibitem[{\citenamefont{Harrow and Low}(2009)}]{harrow2009random}
\bibinfo{author}{\bibfnamefont{A.~W.} \bibnamefont{Harrow}} \bibnamefont{and}
  \bibinfo{author}{\bibfnamefont{R.~A.} \bibnamefont{Low}},
  \bibinfo{journal}{Communications in Mathematical Physics}
  \textbf{\bibinfo{volume}{291}}, \bibinfo{pages}{257} (\bibinfo{year}{2009}).

\bibitem[{\citenamefont{Brand{\~a}o et~al.}(2016)\citenamefont{Brand{\~a}o,
  Harrow, and Horodecki}}]{brandao2016local}
\bibinfo{author}{\bibfnamefont{F.~G.} \bibnamefont{Brand{\~a}o}},
  \bibinfo{author}{\bibfnamefont{A.~W.} \bibnamefont{Harrow}},
  \bibnamefont{and}
  \bibinfo{author}{\bibfnamefont{M.}~\bibnamefont{Horodecki}},
  \bibinfo{journal}{Communications in Mathematical Physics}
  \textbf{\bibinfo{volume}{346}}, \bibinfo{pages}{397} (\bibinfo{year}{2016}).

\bibitem[{\citenamefont{Lund et~al.}(2017)\citenamefont{Lund, Bremner, and
  Ralph}}]{lund2017quantum}
\bibinfo{author}{\bibfnamefont{A.}~\bibnamefont{Lund}},
  \bibinfo{author}{\bibfnamefont{M.~J.} \bibnamefont{Bremner}},
  \bibnamefont{and} \bibinfo{author}{\bibfnamefont{T.}~\bibnamefont{Ralph}},
  \bibinfo{journal}{NPJ Quantum Information} \textbf{\bibinfo{volume}{3}},
  \bibinfo{pages}{1} (\bibinfo{year}{2017}).

\bibitem[{\citenamefont{Harrow and Montanaro}(2017)}]{harrow2017quantum}
\bibinfo{author}{\bibfnamefont{A.~W.} \bibnamefont{Harrow}} \bibnamefont{and}
  \bibinfo{author}{\bibfnamefont{A.}~\bibnamefont{Montanaro}},
  \bibinfo{journal}{Nature} \textbf{\bibinfo{volume}{549}},
  \bibinfo{pages}{203} (\bibinfo{year}{2017}).

\bibitem[{\citenamefont{Terhal and DiVincenzo}(2004)}]{terhal2004adptive}
\bibinfo{author}{\bibfnamefont{B.~M.} \bibnamefont{Terhal}} \bibnamefont{and}
  \bibinfo{author}{\bibfnamefont{D.~P.} \bibnamefont{DiVincenzo}},
  \bibinfo{journal}{Quantum Information \& Computation}
  \textbf{\bibinfo{volume}{4}}, \bibinfo{pages}{134} (\bibinfo{year}{2004}).

\bibitem[{\citenamefont{Aaronson and
  Arkhipov}(2011)}]{aaronson2011computational}
\bibinfo{author}{\bibfnamefont{S.}~\bibnamefont{Aaronson}} \bibnamefont{and}
  \bibinfo{author}{\bibfnamefont{A.}~\bibnamefont{Arkhipov}}, in
  \emph{\bibinfo{booktitle}{Proceedings of the forty-third annual ACM symposium
  on Theory of computing}} (\bibinfo{organization}{ACM}, \bibinfo{year}{2011}),
  pp. \bibinfo{pages}{333--342}.

\bibitem[{\citenamefont{Miller et~al.}(2017)\citenamefont{Miller, Sanders, and
  Miyake}}]{miller2017quantum}
\bibinfo{author}{\bibfnamefont{J.}~\bibnamefont{Miller}},
  \bibinfo{author}{\bibfnamefont{S.}~\bibnamefont{Sanders}}, \bibnamefont{and}
  \bibinfo{author}{\bibfnamefont{A.}~\bibnamefont{Miyake}},
  \bibinfo{journal}{arXiv preprint arXiv:1703.11002}  (\bibinfo{year}{2017}).

\bibitem[{\citenamefont{Bouland et~al.}(2017)\citenamefont{Bouland, Fitzsimons,
  and Koh}}]{bouland2017quantum}
\bibinfo{author}{\bibfnamefont{A.}~\bibnamefont{Bouland}},
  \bibinfo{author}{\bibfnamefont{J.~F.} \bibnamefont{Fitzsimons}},
  \bibnamefont{and} \bibinfo{author}{\bibfnamefont{D.~E.} \bibnamefont{Koh}},
  \bibinfo{journal}{arXiv preprint arXiv:1709.01805}  (\bibinfo{year}{2017}).

\bibitem[{\citenamefont{Bremner et~al.}(2017)\citenamefont{Bremner, Montanaro,
  and Shepherd}}]{bremner2017achieving}
\bibinfo{author}{\bibfnamefont{M.~J.} \bibnamefont{Bremner}},
  \bibinfo{author}{\bibfnamefont{A.}~\bibnamefont{Montanaro}},
  \bibnamefont{and} \bibinfo{author}{\bibfnamefont{D.~J.}
  \bibnamefont{Shepherd}}, \bibinfo{journal}{Quantum}
  \textbf{\bibinfo{volume}{1}}, \bibinfo{pages}{8} (\bibinfo{year}{2017}).

\bibitem[{\citenamefont{Gao et~al.}(2017)\citenamefont{Gao, Wang, and
  Duan}}]{gao2017quantum}
\bibinfo{author}{\bibfnamefont{X.}~\bibnamefont{Gao}},
  \bibinfo{author}{\bibfnamefont{S.-T.} \bibnamefont{Wang}}, \bibnamefont{and}
  \bibinfo{author}{\bibfnamefont{L.-M.} \bibnamefont{Duan}},
  \bibinfo{journal}{Physical Review Letters} \textbf{\bibinfo{volume}{118}},
  \bibinfo{pages}{040502} (\bibinfo{year}{2017}).

\bibitem[{\citenamefont{Hangleiter et~al.}(2017)\citenamefont{Hangleiter,
  Bermejo-Vega, Schwarz, and Eisert}}]{hangleiter2017anti}
\bibinfo{author}{\bibfnamefont{D.}~\bibnamefont{Hangleiter}},
  \bibinfo{author}{\bibfnamefont{J.}~\bibnamefont{Bermejo-Vega}},
  \bibinfo{author}{\bibfnamefont{M.}~\bibnamefont{Schwarz}}, \bibnamefont{and}
  \bibinfo{author}{\bibfnamefont{J.}~\bibnamefont{Eisert}},
  \bibinfo{journal}{arXiv:1706.03786}  (\bibinfo{year}{2017}).

\bibitem[{\citenamefont{Aaronson and Chen}(2016)}]{aaronson2016complexity}
\bibinfo{author}{\bibfnamefont{S.}~\bibnamefont{Aaronson}} \bibnamefont{and}
  \bibinfo{author}{\bibfnamefont{L.}~\bibnamefont{Chen}},
  \bibinfo{journal}{arXiv preprint arXiv:1612.05903}  (\bibinfo{year}{2016}).

\bibitem[{\citenamefont{Knill}(1995)}]{knill1995approximation}
\bibinfo{author}{\bibfnamefont{E.}~\bibnamefont{Knill}},
  \bibinfo{journal}{arXiv:quant-ph/9508006}  (\bibinfo{year}{1995}).

\bibitem[{\citenamefont{Roy and Scott}(2009)}]{roy2009unitary}
\bibinfo{author}{\bibfnamefont{A.}~\bibnamefont{Roy}} \bibnamefont{and}
  \bibinfo{author}{\bibfnamefont{A.~J.} \bibnamefont{Scott}},
  \bibinfo{journal}{Designs, codes and cryptography}
  \textbf{\bibinfo{volume}{53}}, \bibinfo{pages}{13} (\bibinfo{year}{2009}).

\bibitem[{\citenamefont{Toda}(1991)}]{toda1991pp}
\bibinfo{author}{\bibfnamefont{S.}~\bibnamefont{Toda}}, \bibinfo{journal}{SIAM
  Journal on Computing} \textbf{\bibinfo{volume}{20}}, \bibinfo{pages}{865}
  (\bibinfo{year}{1991}).

\bibitem[{\citenamefont{Jones}(1985)}]{jones1985polynomial}
\bibinfo{author}{\bibfnamefont{V.~F.} \bibnamefont{Jones}},
  \bibinfo{journal}{Bulletin of the American Mathematical Society}
  \textbf{\bibinfo{volume}{12}}, \bibinfo{pages}{103} (\bibinfo{year}{1985}).

\bibitem[{\citenamefont{Bordewich et~al.}(2005)\citenamefont{Bordewich,
  Freedman, Lov{\'a}sz, and Welsh}}]{bordewich2005approximate}
\bibinfo{author}{\bibfnamefont{M.}~\bibnamefont{Bordewich}},
  \bibinfo{author}{\bibfnamefont{M.}~\bibnamefont{Freedman}},
  \bibinfo{author}{\bibfnamefont{L.}~\bibnamefont{Lov{\'a}sz}},
  \bibnamefont{and} \bibinfo{author}{\bibfnamefont{D.}~\bibnamefont{Welsh}},
  \bibinfo{journal}{Combinatorics, Probability and Computing}
  \textbf{\bibinfo{volume}{14}}, \bibinfo{pages}{737} (\bibinfo{year}{2005}).

\bibitem[{\citenamefont{Alexander}(1923)}]{alexander1923lemma}
\bibinfo{author}{\bibfnamefont{J.~W.} \bibnamefont{Alexander}},
  \bibinfo{journal}{Proceedings of the National Academy of Sciences}
  \textbf{\bibinfo{volume}{9}}, \bibinfo{pages}{93} (\bibinfo{year}{1923}).

\bibitem[{\citenamefont{Freedman
  et~al.}(2002{\natexlab{a}})\citenamefont{Freedman, Kitaev, and
  Wang}}]{freedman2002simulation}
\bibinfo{author}{\bibfnamefont{M.~H.} \bibnamefont{Freedman}},
  \bibinfo{author}{\bibfnamefont{A.}~\bibnamefont{Kitaev}}, \bibnamefont{and}
  \bibinfo{author}{\bibfnamefont{Z.}~\bibnamefont{Wang}},
  \bibinfo{journal}{Communications in Mathematical Physics}
  \textbf{\bibinfo{volume}{227}}, \bibinfo{pages}{587}
  (\bibinfo{year}{2002}{\natexlab{a}}).

\bibitem[{\citenamefont{Freedman
  et~al.}(2002{\natexlab{b}})\citenamefont{Freedman, Larsen, and
  Wang}}]{freedman2002modular}
\bibinfo{author}{\bibfnamefont{M.~H.} \bibnamefont{Freedman}},
  \bibinfo{author}{\bibfnamefont{M.}~\bibnamefont{Larsen}}, \bibnamefont{and}
  \bibinfo{author}{\bibfnamefont{Z.}~\bibnamefont{Wang}},
  \bibinfo{journal}{Communications in Mathematical Physics}
  \textbf{\bibinfo{volume}{227}}, \bibinfo{pages}{605}
  (\bibinfo{year}{2002}{\natexlab{b}}).

\bibitem[{\citenamefont{Jones}(1983)}]{jones1983braid}
\bibinfo{author}{\bibfnamefont{V.~F.} \bibnamefont{Jones}},
  \bibinfo{journal}{Geometric methods in operator algebras (Kyoto, 1983)}
  \textbf{\bibinfo{volume}{123}}, \bibinfo{pages}{242} (\bibinfo{year}{1983}).

\bibitem[{\citenamefont{Kitaev et~al.}(2002)\citenamefont{Kitaev, Shen, and
  Vyalyi}}]{kitaev2002classical}
\bibinfo{author}{\bibfnamefont{A.~Y.} \bibnamefont{Kitaev}},
  \bibinfo{author}{\bibfnamefont{A.}~\bibnamefont{Shen}}, \bibnamefont{and}
  \bibinfo{author}{\bibfnamefont{M.~N.} \bibnamefont{Vyalyi}},
  \emph{\bibinfo{title}{Classical and quantum computation}},
  vol.~\bibinfo{volume}{47} (\bibinfo{publisher}{American Mathematical Society
  Providence}, \bibinfo{year}{2002}).

\bibitem[{\citenamefont{Shor and Jordan}(2008)}]{shor2008estimating}
\bibinfo{author}{\bibfnamefont{P.~W.} \bibnamefont{Shor}} \bibnamefont{and}
  \bibinfo{author}{\bibfnamefont{S.~P.} \bibnamefont{Jordan}},
  \bibinfo{journal}{Quantum Information \& Computation}
  \textbf{\bibinfo{volume}{8}}, \bibinfo{pages}{681} (\bibinfo{year}{2008}).

\bibitem[{\citenamefont{Alagic et~al.}(2010)\citenamefont{Alagic, Jordan,
  K{\"o}nig, and Reichardt}}]{alagic2010estimating}
\bibinfo{author}{\bibfnamefont{G.}~\bibnamefont{Alagic}},
  \bibinfo{author}{\bibfnamefont{S.~P.} \bibnamefont{Jordan}},
  \bibinfo{author}{\bibfnamefont{R.}~\bibnamefont{K{\"o}nig}},
  \bibnamefont{and} \bibinfo{author}{\bibfnamefont{B.~W.}
  \bibnamefont{Reichardt}}, \bibinfo{journal}{Physical Review A}
  \textbf{\bibinfo{volume}{82}}, \bibinfo{pages}{040302}
  (\bibinfo{year}{2010}).

\bibitem[{\citenamefont{Stockmeyer}(1985)}]{stockmeyer1985approximation}
\bibinfo{author}{\bibfnamefont{L.}~\bibnamefont{Stockmeyer}},
  \bibinfo{journal}{SIAM Journal on Computing} \textbf{\bibinfo{volume}{14}},
  \bibinfo{pages}{849} (\bibinfo{year}{1985}).

\bibitem[{\citenamefont{Brand{\~a}o and
  Horodecki}(2013)}]{brandao2013exponential}
\bibinfo{author}{\bibfnamefont{F.~G.} \bibnamefont{Brand{\~a}o}}
  \bibnamefont{and}
  \bibinfo{author}{\bibfnamefont{M.}~\bibnamefont{Horodecki}},
  \bibinfo{journal}{Quantum Information and Computation}
  \textbf{\bibinfo{volume}{13}}, \bibinfo{pages}{901} (\bibinfo{year}{2013}).

\end{thebibliography}

\end{document}